\documentclass[oneside,english]{amsart}
\usepackage[T1]{fontenc}
\usepackage[latin9]{inputenc}
\usepackage[active]{srcltx}
\usepackage{amsthm}
\usepackage{graphicx}
\usepackage{esint}

\makeatletter
\numberwithin{equation}{section}
\numberwithin{figure}{section}
  \theoremstyle{remark}
  \newtheorem*{rem*}{\protect\remarkname}
  \theoremstyle{definition}
  \newtheorem*{problem*}{\protect\problemname}
\theoremstyle{plain}
\newtheorem{thm}{\protect\theoremname}
  \theoremstyle{plain}
  \newtheorem{prop}[thm]{\protect\propositionname}
  \theoremstyle{definition}
  \newtheorem{defn}[thm]{\protect\definitionname}

\pdfoutput=1
\usepackage{babel}

\makeatother

\usepackage{babel}
  \providecommand{\definitionname}{Definition}
  \providecommand{\problemname}{Problem}
  \providecommand{\propositionname}{Proposition}
  \providecommand{\remarkname}{Remark}
\providecommand{\theoremname}{Theorem}

\begin{document}

\title{Effective diffusion on riemannian fiber bundles}

\author{Carlos Valero Valdes\\
Departamento de Matematicas\\
Universidad de Guanajuato\\
 Guanajuato, Gto, Mexico}

\date{26 January 2015}

\thanks{Partially supported by CONACyT grant 135106. }
\begin{abstract}
The purpose of this paper is to provide equations to model the evolution
of effective diffusion over a Riemannian fiber bundle (under the hypothesis
of infinite diffusion rate along compact fibers). These equations
are obtained by projecting the diffusion equation onto the base manifold
of the fiber bundle. The projection (or dimensional reduction) is
achieved by integrating the diffusion equation along the fibers of
the bundle. This work generalizes an put into a general framework
previous work on effective diffusion over channels and the interfaces
between curved surfaces.
\end{abstract}

\maketitle
\global\long\def\CC{\mathbb{C}}

\global\long\def\RR{\mathbb{R}}

\global\long\def\map{\rightarrow}

\global\long\def\mult{\mathcal{M}}

\global\long\def\EE{\mathcal{E}}

\global\long\def\OO{\mathcal{O}}

\global\long\def\FH{\mathcal{F}}

\global\long\def\CH{\mathcal{H}}

\global\long\def\VV{\mathcal{V}}

\global\long\def\PP{\mathcal{P}}

\global\long\def\CS{\mathcal{C}}

\global\long\def\tangent{T}

\global\long\def\cotangent{\tangent^{*}}

\global\long\def\conormal{\mathcal{C}}

\global\long\def\SS{\mathcal{S}}

\global\long\def\KK{\mathcal{K}}

\global\long\def\NN{\mathcal{N}}

\global\long\def\sym#1{\hbox{S}^{2}#1}

\global\long\def\symz#1{\hbox{S}_{0}^{2}#1}

\global\long\def\proj#1{\hbox{P}#1}

\global\long\def\SO{\hbox{SO}}

\global\long\def\GL{\hbox{GL}}

\global\long\def\U{\hbox{U}}

\global\long\def\tr{\hbox{tr}}

\global\long\def\ZZ{\mathbb{Z}}

\global\long\def\der#1#2{\frac{\partial#1}{\partial#2}}

\global\long\def\dder#1#2{\frac{\partial^{2}#1}{\partial#2^{2}}}

\global\long\def\covder{\hbox{D}}

\global\long\def\diff{d}

\global\long\def\dero#1{\frac{\partial}{\partial#1}}

\global\long\def\ind{\hbox{\, ind}}

\global\long\def\deg{\hbox{deg}}

\global\long\def\smb{S}

\global\long\def\DD{\mathcal{D}}

\global\long\def\QQ{\mathcal{Q}}

\global\long\def\RRR{\mathcal{R}}

\global\long\def\dd#1#2{\frac{d^{2}#1}{d#2^{2}}}

\global\long\def\d#1#2{\frac{d#1}{d#2}}

\global\long\def\and{\hbox{\,\, and\,\,\,}}

\global\long\def\where{\hbox{\,\, where\,\,\,}}

\global\long\def\KK{\mathcal{K}}

\global\long\def\II{\mathcal{I}}

\global\long\def\JJ{\mathcal{J}}

\global\long\def\Re{\hbox{Re}}

\global\long\def\Im{\hbox{Im}}

\global\long\def\nn{\bm{n}}

\global\long\def\div{\hbox{div}}

\global\long\def\lie{\mathcal{L}}

\global\long\def\TT{\mathcal{T}}

\global\long\def\HH{\mathcal{H}}

\section{Introduction}

Understanding spatially constrained diffusion is of fundamental importance
in various sciences, such as biology, chemistry and nano-technology.
However, solving the diffusion equation in arbitrarily constrained
geometries is a very difficult task. One way to tackle it consists
in reducing the degrees of freedom of the problem by considering only
the main direction(s) of transport.
\begin{figure}
\includegraphics[scale=0.45]{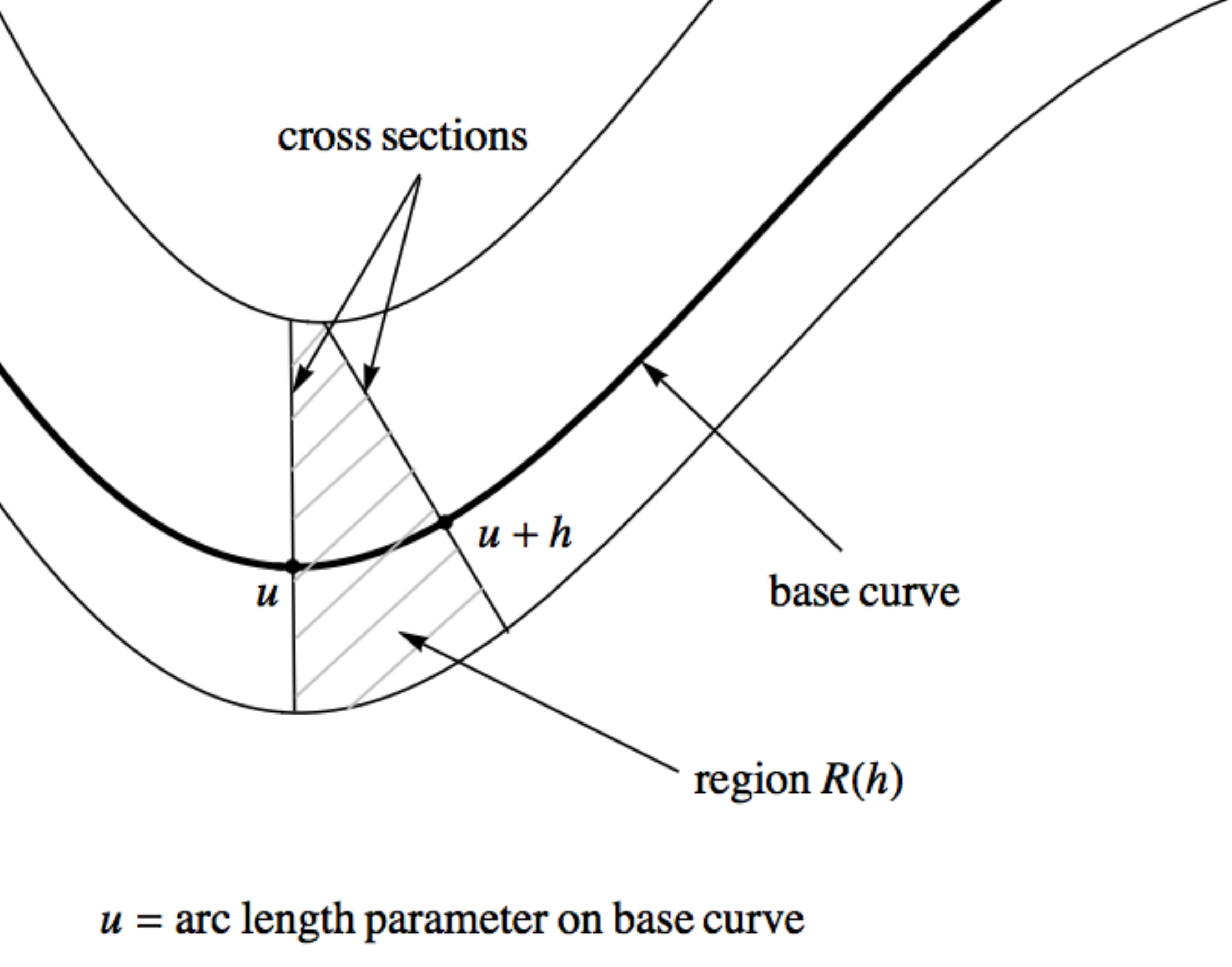}\protect\caption{\label{fig:Rh}Region between to transversal cross section of a channel.}

\end{figure}
For example, the study of diffusion on thin channels can be carried
out by reducing several spatial degrees of freedom to a single one
by means of a projection method. More concretely, consider a process
in a channel modeled by a density function $P$ that obeys the diffusion
equation
\[
\der Pt(x,t)=D_{0}\Delta P(x,t),
\]
subject to the restriction that there is no density flow along the
channel's wall(s). We can construct an \emph{effective density function}
by letting 
\begin{equation}
\rho(u,t)=\lim_{h\mapsto0}\frac{\hbox{Total\, concentration\, of\,\,}P\hbox{\, in\,\,}R(h)}{h},\label{eq:EffectiveConcentrationClassic}
\end{equation}
where the region $R(h)$ is the section of the channel within two
transversal cross sections that are an arc distance of $h$ apart
over a base curve that ``follows'' the channel's geometry (see Figure
\ref{fig:Rh}), and the variable $u$ is the arc-length parameter
on this curve. It turns out that this effective density function $\rho$
obeys in an approximate manner an equation of the form (known as a
generalized Fick-Jacobs equation)

\begin{equation}
\der{\rho}t(u,t)=\dero u\left(\sigma(u)\DD(u)\dero u\left(\frac{\rho(u,t)}{\sigma(u)}\right)\right),\label{eq:GeneralizedFJClassic}
\end{equation}
where $\sigma$ is given by
\begin{equation}
\sigma(u)=\lim_{h\mapsto0}\frac{\hbox{Area}(R(h))}{h}.\label{eq:SigmaClassic}
\end{equation}
The function $\DD$ is known as the \emph{effective diffusion coefficient}
and it encapsulates the effect of the channel's geometry on the diffusion
process along the base curve. Much work has been done (see \cite{fj:antipov,fj:aproximations,fj:bradley,fj:di-projection-diffusion,fj:entropybarrierzwanzig,fj:kp-diffusion-projection,fj:kp-Extendex-fj-variational,fj:kp-fick-jacob-correction,fj:ogawa,fj:projdiff})
to find explicit formulas for $\DD$ in terms of geometrical quantities
associated to the channel, so that the Fick-Jacobs equation models
the evolution of $\rho$ as closely as possible . We can distinguish
two cases.
\begin{enumerate}
\item \emph{Infinite transverse diffusion rate. }In this case it is assumed
that the density function $P$ stabilizes instantly in the transversal
directions of the channel. In mathematical terms this means that $P$
is constant along these transversal directions. This assumption results
in an effective diffusion coefficient that depends on the curvature
function of the base curve and $0$-th order geometrical quantities
of the cross sections, such as width or area.
\item \emph{Finite transverse diffusion rate. }In this case the finite time
of transversal stabilization of $P$ is taken into account. This is
characterized mathematically in that the resulting formulas for $\DD$
involve the curvature function of the base curve, and tangential and
curvature information of the channel's wall(s).
\end{enumerate}
The selection of the base curve is very important in the dimension-reduction
technique described  above. This is demonstrated by the fact (see
\cite{fj:projdiff,fj:ogawa}) that if for channels of constant width
the base curve is chosen properly, then the formulas for $\DD$ coincide
for the finite and infinite transversal diffusion rate cases. 

Motivated by the above discussion, we know describe the main purpose
of the paper. We develop a very general theory, in the infinite transversal
diffusion rate case, for projecting the diffusion equation in a space
of dimension $n$ to a base space of dimension $m$ with $m<n$. We
do this in the context of the theory of \emph{fibre bundles}. Such
objects have a total space $E$, a base space $M$ and a projection
map $\pi:E\rightarrow M$. In the case discussed above, the total
space $E$ is the channel, $M$ is the base curve, and $\pi$ sends
points on the transversal cross sections (fibers) to their base point
in the base curve. The process of passing from the density $P$ to
the reduced density $\rho$ is a particular case of a very well known
construction: that of integrating a differentiable form along the
fibers of the bundle. Using these tools we are able to give global
and coordinate-free proofs of all our results. In this general setting,
the effective diffusion $\DD$ becomes an endomorphism of the tangent
bundle of $M$, i.e for every $x$ in $M$ we have that $\DD(x)$
is a linear map in the tangent space $\tangent_{x}M$ of $M$. We
will compute the effective diffusion using local frames instead of
local coordinates, which results in a simpler and more geometric way
of doing calculations. The paper is organized as follows
\begin{itemize}
\item In section \ref{sec:EffectiveDiffusionEquation}, we show how to reduce
the continuity equation in the fiber bundle $E$ to a reduced continuity
equation in its base space $M$. We prove that if Fick's law holds
on $E$ (for a constant diffusion coefficient $D_{0}$) and if we
have infinite diffusion rate in fiber direction, then the reduced
continuity equation becomes a diffusion equation in $M$ (see Proposition
\ref{prop:EffectiveDiffusionEquation}). This last equation involves
an effective diffusion coefficient $\DD$, which is a bundle endomorphism
of the tangent space of $M$.
\item In section \ref{sec:CurvesOnThePlane} we compute the effective diffusion
for channels of constant width over arbitrary curves on the plane.
We have obtained this result previously in \cite{fj:projdiff}, and
our re-derivation of the formula serves as a test case of our general
theory.
\item In section \ref{sec:SurfacesIn3dSpace} we compute the effective diffusion
endomorphism $\DD$ corresponding to the interface of two equidistant
surfaces in 3-dimensional space. We do this by showing that the principal
directions of the base surface are eigenvectors of $\DD$, and then
computing the eigenvalues of $\DD$ in these directions. Using this
result, we show that the calculation of $\DD$ made by Ogawa in \cite{fj:ogawa-surfaces}
for an elliptical cylinder, is just an approximation to ours obtained
by only considering the two first terms in a series expansion of the
function $\hbox{arctanh}$ (the inverse of the hyperbolic tangent).
\item In section \ref{sec:EffectiveDiffSurfaceTube} we compute the effective
diffusion on the surface (not the interior) of a circular channel
over an arbitrary curve in $\RR^{3}$. The main point of the calculation
is to illustrate how our techniques still apply to fibre bundles whose
fibers are manifolds without boundary (in this case, circles).
\item In the Appendix we make a brief review of the geometrical concepts
needed for the construction of the effective diffusion endomorphism.
\end{itemize}

\section{The effective diffusion equation}

\label{sec:EffectiveDiffusionEquation}

Let $E$ be fiber bundle over an $m$-dimensional manifold $M$, having
compact fibers (with or without boundary) of dimension $k$, and projection
map $\pi:E\rightarrow M$. We will assume that $M$ and $E$ are orientable,
and with Riemannian metrics $<,>_{E}$ and $<,>_{M}.$

\subsection*{The continuity equation}

The \emph{continuity equation }on the fibre bundle $E$ is given by\emph{
}
\begin{equation}
\der Pt+\div(J)=0,\label{eq:TransportEquation}
\end{equation}
where the \emph{density function} $P$ is a time dependent function
on $E$, and the \emph{density flow} $J$ is a time dependent vector
field in $E$. The divergence of $J$ is given by
\begin{equation}
\div(J)=(*(d(*J^{\flat})))^{\sharp},\label{eq:divergence}
\end{equation}
where $d$ is the differential operator acting on differentiable forms
in $E$, $*$ is the Hodge star operator, and the $\flat$-operator
converts vector fields to 1-forms (see the Appendix).

\subsection*{Integrating the continuity equation along the fibers}

Since we are assuming that the fibers of $E$ are compact, by integrating
along them we can define the operator $\pi_{*}$ sending $l$-forms
in $E$ to $(l-k$)-forms in $M$. For any $l-$form $\omega$ on
$E$ we define
\[
\pi_{*}(\omega)_{x}(X_{1},\ldots,X_{l-k})=\int_{\pi^{-1}(x)}\beta_{\omega}
\]
where 
\[
\beta_{\omega}(Y_{1},\ldots,Y_{k})=\omega(\tilde{X}_{1},\ldots,\tilde{X}_{l-k},Y_{1},\ldots,Y_{k})
\]
and $\tilde{X}_{i}$ is a lift of $X_{i}$, i.e $D\pi(\tilde{X}_{i})=X_{i}$.
Observe that $\tilde{X}_{i}$ is defined on the whole fiber $\pi^{-1}(x).$

By applying $*$ to \ref{eq:TransportEquation} and using formula
\ref{eq:divergence} we get
\begin{equation}
\dero t(*P)+d(*J^{\flat})=0.\label{eq:TransportEqDiffForm}
\end{equation}
If we apply $\pi_{*}$ and then $*$ to the above equation we obtain
\begin{equation}
\der{\rho}t+*(\pi_{*}(d(*J^{\flat})))=0,\label{eq:EffectiveTransportEq0}
\end{equation}
where
\begin{equation}
\rho=*(\pi_{*}(*P))=\frac{\pi_{*}(P\mu_{E})}{\mu_{M}},\label{eq:effectiveConcentration}
\end{equation}
and $\mu_{E}$ and $\mu_{M}$ be the metric volume forms in $E$ and
$M$. 
\begin{rem*}
Formula \ref{eq:effectiveConcentration} is a generalization of formula
\ref{eq:EffectiveConcentrationClassic}.
\end{rem*}

\subsection*{Neumann Boundary conditions}

We will need adequate boundary condition on $J$ in order to write
equation \ref{eq:EffectiveTransportEq0} as a continuity equation
in $M$. If the fibers of $E$ are \emph{manifolds without boundary}
then $\pi_{*}$ commutes with $d$ (see \cite[pg. 62]{kn:bott}),
i.e for any differential form $\omega$ we have that 
\begin{equation}
d(\pi_{*}\omega)=\pi_{*}(d\omega)\label{eq:pistardcommute}
\end{equation}
When the fibers of $E$ are \emph{manifolds with boundary} we will
assume that $\omega$ vanishes on all the vectors perpendicular to
the boundary $\partial E$ of $E$. This last condition is equivalent
to the assuming that there is no density flow across $\partial E$,
i.e $J$ is parallel to the boundary of $E$. This last condition
ensures that formula \ref{eq:pistardcommute} still holds.

\subsection*{Integrating the density flow along the fibers }

We will now use formula \ref{eq:pistardcommute} to show that we can
write \ref{eq:EffectiveTransportEq0} as a continuity equation in
$M$. We need to find a time dependent vector field $j$ in $M$ such
that
\[
\div(j)=*(\pi_{*}(d(*J^{\flat}))).
\]
Using formula \ref{eq:pistardcommute}, and the definition of the
divergence of $j$, we can write the above formula as
\[
*(d(*j^{\flat}))=*(d(\pi_{*}(*J^{\flat}))).
\]
This last equation is satisfied if we let
\begin{equation}
j=(-1)^{m-1}(*(\pi_{*}(*J^{\flat})))^{\sharp},\label{eq:effectiveFlow}
\end{equation}
where the $\sharp$-operator converts 1-forms to vector fields (see
Appendix).

\subsection*{The effective continuity equation}

We will refer to the time dependent function $\rho$ given by formula
\ref{eq:effectiveConcentration} as the \emph{effective density} function,
and to the time dependent vector field $j$ given by \ref{eq:effectiveFlow}
as the \emph{effective density flow}. We proved above that these objects
satisfy the equation
\begin{equation}
\der{\rho}t+\div(j)=0,\label{eq:EffectiveContinuity}
\end{equation}
which we will refer to as the \emph{effective continuity equation}.

\subsection*{The diffusion equation}

Fick's law establishes that
\[
J(x,t)=-D(x)\nabla P(x,t),
\]
where for $x$ in $E$ we have that $D(x)$ is a linear operator from
$\tangent_{x}E$ to $\tangent_{x}E$, i.e an endomorphism of $\tangent E_{x}$.
The simplest choice of $D$ is to let it be scalar multiplication
by a constant $D_{0}$. By using more general $D$'s we can model
the inhomogeneity and anisotropy of $E$. Assuming Fick's law, the
continuity equation in $E$ becomes the \emph{diffusion equation}
\begin{equation}
\der Pt(x,t)=\div(D(x)\nabla P(x,t)),\label{eq:diffEq}
\end{equation}
In our work we will always assume that the endomorphism $D$ on $\tangent E$
is multiplication by a positive scalar $D_{0}$, in which case the
diffusion equation becomes 

\[
\der Pt(x,t)=D_{0}\Delta P(x,t),
\]
where the \emph{laplacian operator} $\Delta$ applied to $P$ is
\[
\Delta P=\div(\nabla P).
\]

\begin{problem*}
If Fick's law holds in $E$ for $D=D_{0}$, does the effective density
$\rho$ obeys a diffusion equation in $M$?
\end{problem*}
There is an important case when we can answer the above question positively
(see Proposition \ref{prop:EffectiveDiffusionEquation} in the next
paragraph).

\subsection*{An effective diffusion equation for infinite fiber diffusion rate}

If the fibers of $E$ are ``small enough'' compared to the ``size''
of $M$, then it is to be expected that the density $P$ will stabilize
faster along the fibers than along $M$. If we assume that this stabilization
occurs infinitely fast, then $P$ must be constant along the fibers
of $E$. We borrow the nomenclature from the physics literature, and
refer to this situation by saying that there is an \emph{infinite
fiber diffusion rate}. Under this assumption, can write
\[
P=\pi^{*}Q=Q\circ\pi
\]
for a time dependent function $Q$ in $M$. Using \ref{eq:effectiveConcentration}
and the projection formula (see \cite[pg. 63]{kn:bott}) , we obtain
\[
\rho=\frac{\pi_{*}(\pi^{*}(Q)\mu_{E})}{\mu_{M}}=\frac{Q\pi_{*}(\mu_{E})}{\mu_{M}},
\]
which allows us to obtain $Q$ in terms of $\rho$ as 
\begin{equation}
Q=\frac{\rho}{\sigma},\label{eq:Pfromp}
\end{equation}
where
\begin{equation}
\sigma=\frac{\pi_{*}(\mu_{E})}{\mu_{M}}.\label{eq:sigma}
\end{equation}
Observe that if $R$ is a region in $M$ then 
\begin{equation}
\int_{R}\sigma\mu_{M}=\int_{\pi^{-1}(R)}\mu_{E}=\hbox{vol}_{E}(\pi^{-1}(R)).\label{eq:sigmaAsVolumes}
\end{equation}

\begin{rem*}
Formula \ref{eq:sigma} is a generalization of that given by \ref{eq:SigmaClassic}.
\end{rem*}
If Fick's law holds in $E$ for $D=D_{0}$, then we have
\[
J^{\flat}=-D_{0}d(\pi^{*}Q)=-D_{0}\pi^{*}(dQ),
\]
where $\pi^{*}$ is the pull-back of forms under $\pi$. Hence, we
can write the effective flow \ref{eq:effectiveFlow} as
\begin{equation}
j=-(-1)^{m-1}D_{0}(*(\pi_{*}(*(\pi^{*}(dQ)))))^{\sharp}.\label{eq:EffectiveFlowGradientCase}
\end{equation}
If we define $\DD$ by 
\begin{equation}
\DD=(-1)^{m-1}\left(\frac{D_{0}}{\sigma}\right)(\sharp\circ*\circ\pi_{*}\circ*\circ\pi^{*}\circ\flat),\label{eq:AbstractDOperator}
\end{equation}
where $\circ$ stands for composition of operators, then by combining
formulas \ref{eq:Pfromp} and \ref{eq:EffectiveFlowGradientCase}
we obtain
\begin{equation}
j=-\sigma\DD(\nabla Q)=-\sigma\DD\left(\nabla\left(\frac{\rho}{\sigma}\right)\right).\label{eq:FluxInTermsOfD}
\end{equation}
The above formulas now lead to the following result.
\begin{prop}
\label{prop:FickJacobs}For the case of infinite fiber diffusion rate,
the effective density function $\rho$ satisfies the equation 
\begin{equation}
\der{\rho}t(x,t)=\div\left(\sigma(x)\DD(x)\left(\nabla\left(\frac{\rho(x,t)}{\sigma(x)}\right)\right)\right),\label{eq:SuperGeneralizedFJ}
\end{equation}
where $\DD:\tangent X\rightarrow\tangent X$ is the vector bundle
morphism defined in formula \ref{eq:AbstractDOperator}.\end{prop}
\begin{proof}
This follows from the effective continuity equation \ref{eq:EffectiveContinuity}
and formula \ref{eq:FluxInTermsOfD}. As defined by formula \ref{eq:AbstractDOperator},
$\DD$ sends vector fields in $M$ to vector fields in $M$. To prove
that for any $x$ in $M$ it is in fact linear map $\DD(x):\tangent_{x}M\rightarrow\tangent_{x}M$,
we show that $\DD(fX)=f\DD(X)$ for any smooth function $f:M\rightarrow\RR$
and vector field $X$ in $M$. We have that
\[
(*\circ\pi^{*}\circ\flat)(fX)=(\pi^{*}f)(\pi*(X^{\flat})).
\]
Using this and the projection formula (see \cite[pg. 63]{kn:bott}),
we obtain that
\[
\pi_{*}(*\circ\pi^{*}\circ\flat)(fX)=f\pi_{*}(\pi*(X^{\flat})),
\]
and hence
\[
\DD(fX)=f\DD(X).
\]
\end{proof}
\begin{defn}
We will refer to the endomorphism $\DD$ of $\tangent M$, given by
formula \ref{eq:AbstractDOperator}, as the\emph{ effective diffusion
endomorphism.}\end{defn}
\begin{rem*}
Equation \ref{eq:SuperGeneralizedFJ} is a generalization of \ref{eq:GeneralizedFJClassic}.\end{rem*}
\begin{prop}
\label{prop:EffectiveDiffusionEquation}For the case of infinite fiber
diffusion rate, there exists a metric in $M$ such that the effective
density function $\rho$ satisfies the diffusion equation
\begin{equation}
\der{\rho}t(x,t)=\div\left(\DD(x)\nabla\rho(x,t)\right),\label{eq:effectiveDiffEquation}
\end{equation}
where $\DD$ is the effective diffusion endomorphism. Furthermore,
this choice of metric is such that for any region $R$ in $M$ we
have that
\[
\hbox{vol}_{M}(R)=\hbox{vol}_{E}(\pi^{-1}(R)).
\]
\end{prop}
\begin{proof}
Let $<,>_{E}$ and $<,>_{M}$ be any metrics in $E$ and $M$. Define
a new metric $<,>'_{M}$ in $M$ by the formula
\[
<X,Y>'_{M}=\sigma^{2/m}<X,Y>_{M},
\]
for $\sigma$ defined in \ref{eq:sigma}. Then, since
\[
\sigma^{2/m}=\left(\frac{\pi_{*}(\mu_{E})}{\mu_{M}}\right)^{2/m}
\]
we have that
\[
\sigma'_{M}=\frac{\pi_{*}(\mu_{E})}{\mu'_{M}}=\frac{\pi_{*}(\mu_{E})}{(\sigma^{2/m})^{m/2}\mu_{M}}=1.
\]
Equation \ref{eq:effectiveDiffEquation} then follows directly from
Proposition \ref{prop:FickJacobs}, and the last part of the Proposition
follows from equation \ref{eq:sigmaAsVolumes}.
\end{proof}

\section{Constant-width channels on the plane}

\label{sec:CurvesOnThePlane}

\begin{figure}

\includegraphics[scale=0.4]{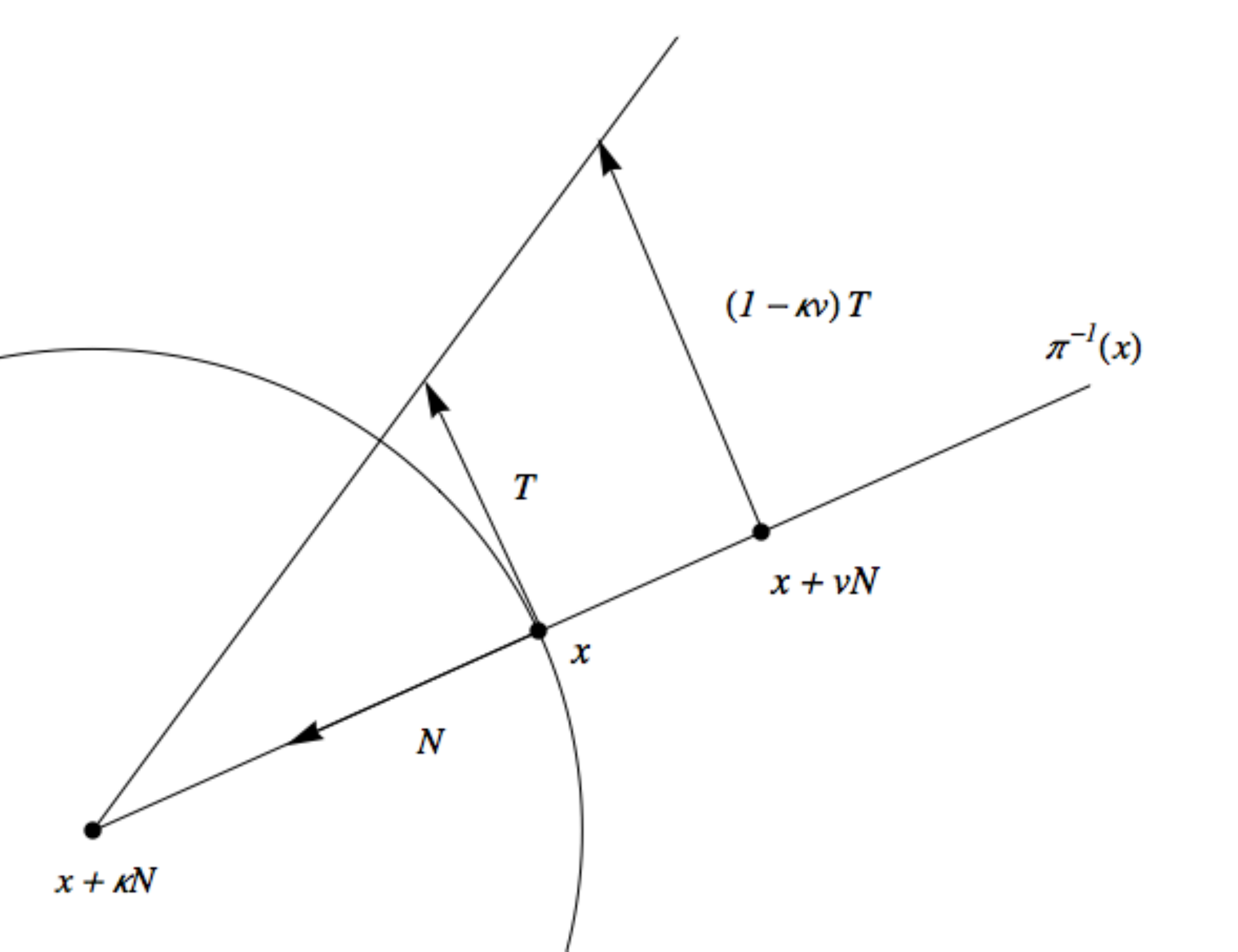}\protect\caption{\label{fig:Lift2D}Lifting tangent vector $T$ to vector $(1-\kappa v)T$
along a fiber $\pi^{-1}(x)$}

\end{figure}

In this section we compute the effective diffusion function of a channel
of constant width $w$ over a curve $C$ on the plane. Such a channel
can be represented as the set 
\[
E=\{x+vN(x)|x\in C\hbox{\, and\,}-w/2\leq v\leq w/2\},
\]
where $N$ is a unit normal field to $C$. For a reasonable curves
(e.g compact) and small $w$, the space $E$ is a fibre bundle with
projection map $\pi:E\rightarrow M$ given by $\pi(p)=x$ where
\[
p=x+vN(x).
\]
Let $T$ be a unit tangent field to $C$ that makes the frame $T,N$
positively oriented. To perform integration over the fibers of $E$
we need to compute the lift of $T$ to $E$. If we define 
\begin{eqnarray*}
\TT(p) & = & (1-\kappa(x)v)T(x)\\
\NN(p) & = & N(x)
\end{eqnarray*}
then $\TT$ is such a lift of $T$ (see Figure \ref{fig:Lift2D}),
and $D\pi(\NN)=0$. Let $T^{*}$ be the dual field to $T$, and $\TT^{*},\NN^{*}$
be the dual frame to $\TT,\NN$. The matrix of $<,>_{E}$ in the frame
$\TT,\NN$ is given by 
\[
g=\left(\begin{array}{cc}
(1-\kappa v)^{2} & 0\\
0 & 1
\end{array}\right).
\]

\subsection*{Computing $\sigma$}

The volume element in $E$ is 
\[
\mu_{E}=(1-\kappa v)\TT^{*}\wedge\NN^{*},
\]
and hence
\[
\pi_{*}(\mu_{E})=\left(\int_{-w/2}^{w/2}(1-\kappa v)dv\right)T^{*}=wT^{*}.
\]
Since the volume form in $C$ is $T^{*},$ we conclude that
\[
\sigma=\frac{wT^{*}}{T^{*}}=w.
\]

\subsection*{Computing $\protect\DD$}

Observe that
\[
\pi^{*}(T^{*})=\TT^{*},
\]
and by formula \ref{eq:StarOneForm} in the Appendix we have that
\[
*(\TT^{*})=g^{11}\det(g)^{1/2}\NN^{*}=(1-\kappa v)^{-1}\NN^{*}.
\]
Hence
\begin{eqnarray*}
\pi_{*}(*(\pi^{*}(T^{*}))) & = & \int_{-w/2}^{w/2}(1-\kappa v)^{-1}dv\\
 & = & \frac{1}{\kappa}\log\left(\frac{1+\kappa w/2}{1-\kappa w/2}\right)
\end{eqnarray*}
Using formula \ref{eq:AbstractDOperator} we obtain
\begin{eqnarray}
\DD & = & \frac{D_{0}}{\kappa w}\log\left(\frac{1+\kappa w/2}{1-\kappa w/2}\right)\nonumber \\
 & = & \left(\frac{2D_{0}}{\kappa w}\right)\hbox{arctanh}(\kappa w/2).\label{eq:curveD}
\end{eqnarray}
We obtained this formula in \cite{fj:projdiff} by different methods.

\section{The interface between two equidistant surfaces in 3-d space}

\label{sec:SurfacesIn3dSpace}

Let $S$ be an orientable surface in $\RR^{3}$ and $N$ a unit normal
field to this surface. For small $w>0$ and a ``reasonable surface''
the space 
\begin{equation}
E=\{x+vN(x)|x\in S\and-w/2\leq v\leq w/2\}.\label{eq:surfaceWithThickness}
\end{equation}
is a fibre bundle over $S$ with projection map $\pi:E\rightarrow M$
given by $\pi(p)=x,$ where
\[
p=x+vN(x).
\]
The surface $S$ has principal directions fields $T_{1}$ and $T_{2}$,
with corresponding principal curvatures $\kappa_{1}$ and $\kappa_{2}$.
If we define
\begin{eqnarray*}
\TT_{1}(p) & = & (1-\kappa_{1}(x)v)T_{1},\\
\TT_{2}(p) & = & (1-\kappa_{2}(x)v)T_{2},\\
\NN(p) & = & N(x),
\end{eqnarray*}
then $\TT_{1}$ and $\TT_{2}$ are lifts of $T_{1}$ and $T_{2}$,
and $D\pi(\NN)=0$. The metric in $<,>_{E}$ is represented in the
$\TT_{1},\TT_{2},\NN$ frame by 
\[
g=\left(\begin{array}{ccc}
(1-\kappa_{1}v)^{2} & 0 & 0\\
0 & (1-\kappa_{2}v)^{2} & 0\\
0 & 0 & 1
\end{array}\right).
\]
We will let $\TT_{1}^{*},\TT_{2}^{*},\NN^{*}$ be the dual frame to
$\TT_{1},\TT_{2},\NN$.

\subsection*{Computing $\sigma$}

The volume form in $E$ is 
\[
\mu_{E}=(1-\kappa_{1}v)(1-\kappa_{2}v)\TT_{1}^{*}\wedge\TT_{2}^{*}\wedge\NN^{*},
\]
and hence
\begin{eqnarray*}
\pi_{*}(\mu_{E}) & = & \left(\int_{-w/2}^{w/2}(1-\kappa_{1}v)(1-\kappa_{2}v)dv\right)T_{1}^{*}\wedge T_{2}^{*},\\
 & = & w(1+\kappa_{1}\kappa_{2}w^{2}/12)T_{1}^{*}\wedge T_{2}^{*}.
\end{eqnarray*}
Since the volume form in $M$ is $T_{1}^{*}\wedge T_{2}^{*}$, we
conclude from formula \ref{eq:sigma} that 
\[
\sigma=w(1+\kappa_{1}\kappa_{2}w^{2}/12).
\]

\subsection*{Computing $\protect\DD$}

Using formula \ref{eq:StarOneForm} in the Appendix we obtain that
\begin{eqnarray*}
*(\TT_{1}^{*}) & = & g^{11}\det(g)^{1/2}\TT_{2}^{*}\wedge\NN^{*}=\left(\frac{1-\kappa_{2}v}{1-\kappa_{1}v}\right)\TT_{2}^{*}\wedge\NN^{*},\\
*(\TT_{2}^{*}) & = & -g^{22}\det(g)^{1/2}\TT_{1}^{*}\wedge\NN^{*}=-\left(\frac{1-\kappa_{1}v}{1-\kappa_{2}v}\right)\TT_{1}^{*}\wedge\NN^{*}.
\end{eqnarray*}
From these formulas and the identities $*T_{1}^{*}=T_{2}^{*},*T_{2}^{*}=-T_{1}^{*}$,
we obtain
\begin{eqnarray*}
*(\pi_{*}(*(\TT_{1}^{*}))) & =- & \left(\int_{-w/2}^{w/2}\left(\frac{1-\kappa_{2}v}{1-\kappa_{1}v}\right)dv\right)T_{1}^{*},\\
*(\pi_{*}(*(\TT_{2}^{*}))) & =- & \left(\int_{-w/2}^{w/2}\left(\frac{1-\kappa_{1}v}{1-\kappa_{2}v}\right)dv\right)T_{2}^{*}.
\end{eqnarray*}
By evaluating the above integrals, using formulas
\[
\pi*(T_{1}^{*})=\TT_{1}^{*},\pi*(T_{2}^{*})=\TT_{2}^{*},
\]
and the fact that the matrix of the metric $<,>_{M}$ is the identity,
formula \ref{eq:AbstractDOperator} yields
\[
\DD(T_{1})=\DD_{1}T_{1}\and\DD(T_{2})=\DD_{2}T_{2}.
\]
 where
\begin{eqnarray}
\DD_{1} & = & \frac{D_{0}\left(w\kappa_{1}\kappa_{2}+2(\kappa_{2}-\kappa_{1})\hbox{arctanh}(\kappa_{1}w/2)\right)}{\kappa_{1}^{2}w(1+\kappa_{1}\kappa_{2}w^{2}/12)},\label{eq:SurfaceD1}\\
\DD_{2} & = & \frac{D_{0}\left(w\kappa_{1}\kappa_{2}-2(\kappa_{2}-\kappa_{1})\hbox{arctanh}(\kappa_{2}w/2)\right)}{\kappa_{2}^{2}w(1+\kappa_{1}\kappa_{2}w^{2}/12)}.\label{eq:SurfaceD2}
\end{eqnarray}
We have just proved the following result.
\begin{prop}
\label{prop:EigenD}Let $S$ be a surface in $\RR^{3}$ with principal
direction fields $T_{1}$ and $T_{2}$, and corresponding principal
curvatures $\kappa_{1}$ and $\kappa_{2}$. For the bundle $E$ given
by \ref{eq:surfaceWithThickness}, we have that the fields $T_{1}$
and $T_{2}$ are eigenvectors of $\DD$ with corresponding eigenvalues
$\DD_{1}$ and $\DD_{2}$ given \ref{eq:SurfaceD1} and \ref{eq:SurfaceD2}.\end{prop}
\begin{rem*}
At an umbilical point $x$ of $S$ (i.e where $\kappa_{1}(x)=\kappa_{2}(x))$
the above proposition is still valid for any pair of orthonormal vectors
$T_{1},T_{2}$ in $\tangent_{x}S$.
\end{rem*}
It is important to observe that if we want to write the effective
diffusion equation \ref{eq:effectiveDiffEquation} in coordinates,
we need to express the principal direction fields $T_{1}$ and $T_{2}$
in terms of the corresponding coordinate fields. The reason for this
is that the divergence operator that enters into the effective diffusion
equation needs coordinates for its computation. 
\begin{rem*}
We can express eigenvalues of $\DD$ in terms of the gaussian and
mean curvatures
\[
K=\kappa_{1}\kappa_{2}\and H=\frac{1}{2}(\kappa_{1}+\kappa_{2})
\]
by using the identities
\[
\kappa_{1}=H+\sqrt{H^{2}-K}\and\kappa_{2}=H-\sqrt{H^{2}-K}.
\]

\end{rem*}
We now discuss some applications of Proposition \ref{prop:EigenD}
to specific families of surfaces.

\subsection*{Spheres}

In this case we have that the principal curvatures $\kappa_{1}$ and
$\kappa_{2}$ satisfy
\[
\kappa_{1}=\kappa_{2}=1/r,
\]
where $r$ is the radius of the sphere. From formulas \ref{eq:SurfaceD1}
and \ref{eq:SurfaceD2}, we obtain that the eigenvalues of $\DD$
are
\begin{equation}
\DD_{1}=\DD_{2}=\frac{D_{0}12r^{2}}{12r^{2}+w^{2}}=\frac{D_{0}}{1+\frac{w^{2}}{12r^{2}}}.\label{eq:D1D2Sphere}
\end{equation}

\subsection*{Recovering the one dimensional case}

Let $C$ be any curve on the plane with curvature function $\kappa$,
and let $S=C\times\RR\subset\RR^{3}$. The principal fields of $S$
are the unit tangent $T$ to the curve and $(0,0,1)$, with corresponding
eigenvalues $\kappa$ and $0$. Using formulas \ref{eq:SurfaceD1}
and \ref{eq:SurfaceD2} we obtain
\begin{eqnarray*}
\DD_{1} & = & \left(\frac{2D_{0}}{\kappa w}\right)\hbox{arctanh}(\kappa w/2),\\
\DD_{2} & = & D_{0}.
\end{eqnarray*}
Hence, the eigenvalue $\DD_{1}$ of $\DD$ coincides with the case
of curves in the plane discussed in section \ref{sec:CurvesOnThePlane}
(see Formula \ref{eq:curveD}). 

We will now show how to write equation \ref{eq:SuperGeneralizedFJ}
in local coordinates. Consider the coordinates $(s,z)$ in $S$ where
$s$ is the arc-length parameter of the curve $C$ and $z$ is standard
$z$-coordinate in $\RR^{3}$. The coordinate fields
\[
\dero s\and\dero z
\]
are the principal direction fields of $S$, and 
\begin{eqnarray*}
\DD\left(\dero s\right) & = & \left(\frac{2D_{0}\hbox{arctanh}(\kappa w/2)}{\kappa w}\right)\dero s,\\
\DD\left(\dero z\right) & = & D_{0}\dero z.
\end{eqnarray*}
Since the metric matrix is the identity, we have that
\[
\nabla\rho=\der{\rho}s\dero s+\der{\rho}z\dero z.
\]
We conclude that in this case formula \ref{eq:SuperGeneralizedFJ}
becomes 
\begin{equation}
\der{\rho}t=\left(\frac{2D_{0}\hbox{arctanh}(\kappa w/2)}{\kappa w}\right)\frac{\partial^{2}\rho}{\partial s^{2}}+\frac{\partial^{2}\rho}{\partial z^{2}}.\label{eq:FJLinedSurface}
\end{equation}
If we use the series expansion 
\[
\hbox{arctanh}(x)=x+\frac{x^{3}}{3}+\frac{x^{5}}{5}+\ldots,
\]
we obtain
\[
\frac{2D_{0}\hbox{arctanh}(\kappa w/2)}{\kappa w}=1+\frac{w^{2}\kappa^{2}}{12}+\frac{w^{4}\kappa^{4}}{180}+\cdots
\]
If we use only the first two terms of the above series in equation
\ref{eq:FJLinedSurface} we obtain
\[
\der{\rho}t=\left(1+\frac{w^{2}\kappa^{2}}{12}\right)\frac{\partial^{2}\rho}{\partial s^{2}}+\frac{\partial^{2}\rho}{\partial z^{2}}.
\]
This last formula is that obtained by Ogawa in \cite{fj:ogawa-surfaces}.

\subsection*{The torus}

\begin{figure}

\includegraphics[scale=0.45]{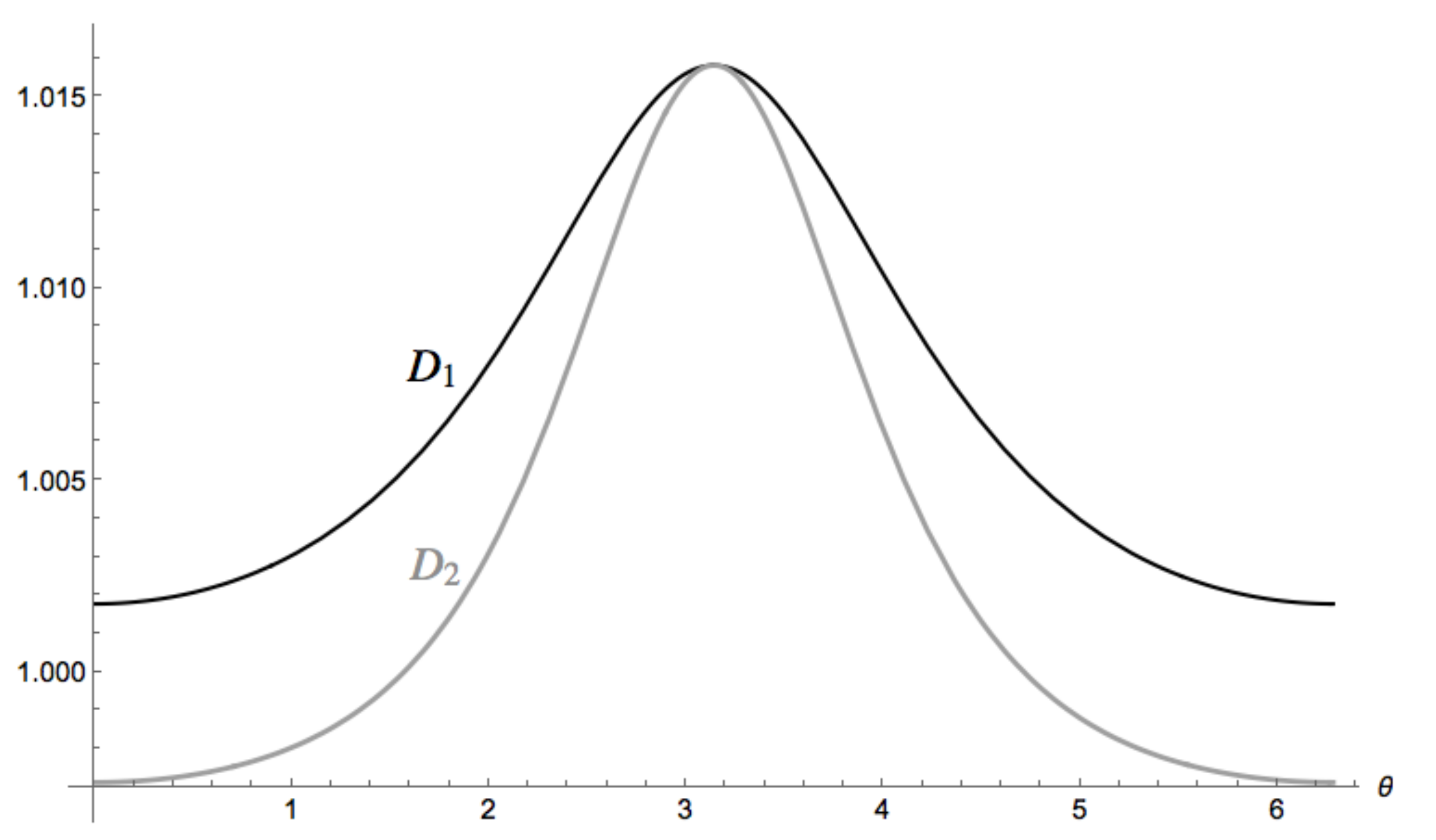}\protect\caption{\label{fig:torusD1}Eigenvalues $\protect\DD_{1}$ and $\protect\DD_{2}$
of the effective diffusion $\protect\DD$ for $D_{0}=1$, on a torus
with inner radius $r=1$ , outer radius $R=2$ and width $w=1/4$. }
\end{figure}
For a torus with inner radius $r$ and outer radius $R$ we have that
\[
\kappa_{1}=-1/r\and\kappa_{2}=-\frac{\cos(\theta)}{R+r\cos(\theta)},
\]
where $\theta$ is the variable parametrizing the parallels of the
torus. In Figure \ref{fig:torusD1} we show the graphs of $\DD_{1}$
and $\DD_{2}$ obtained by using the above values of $\kappa_{1}$
and $\kappa_{2}$ for specific values of $r,R$ and $w$.

\section{Effective diffusion in the surface of a tube}

\label{sec:EffectiveDiffSurfaceTube}

Let $C$ be a curve in three dimensional space, and let $T,N,B$ be
the corresponding Serret-Frenet frame. We will let $E$ be the set
of points of the form (for a constant radius $r$)
\[
p=x+r\cos(\theta)N+r\sin(\theta)B\where{x\in C}.
\]
We want to construct a lift of $T$ to $E$. To do this, consider
$x=x(s),T=T(s),N=N(s)$ and $B=B(s)$ as functions of the arc length
parameter $s$ of $C$. From the formula
\[
\pi(p(s))=x(s)
\]
we obtain
\[
D\pi\left(\frac{dp}{ds}\right)=T\and D\pi\left(\frac{dp}{d\theta}\right)=0,
\]
where (using the Serret-Frenet formulas)
\begin{eqnarray*}
\frac{dp}{ds} & = & (1-\kappa r\cos(\theta))T+\tau r(-\sin(\theta)N+\cos(\theta)B),\\
\frac{dp}{d\theta} & = & -r\sin(\theta)N+r\cos(\theta)B,
\end{eqnarray*}
for $\kappa$ and $\tau$ the curvature and torsion of $C$. Hence,
if we define
\begin{eqnarray*}
\TT & = & (1-\kappa r\cos(\theta))T,\\
\NN & = & -r\sin(\theta)N+r\cos(\theta)B,
\end{eqnarray*}
then $\TT$ is lift of $T$ and $D\pi(\NN)=0$. The metric matrix
in the $\TT,\NN$ frame is 
\[
g=\left(\begin{array}{cc}
(1-\kappa r\cos(\theta))^{2} & 0\\
0 & r^{2}
\end{array}\right).
\]

\subsection*{Computing $\sigma$}

The volume form in $E$ is 
\[
\mu_{E}=r(1-\kappa r\cos(\theta))\TT^{*}\wedge\HH^{*},
\]
so that
\[
\pi_{*}(\mu_{E})=\left(\int_{0}^{2\pi}r(1-r\kappa\cos(\theta))d\theta\right)T^{*}=2\pi rT^{*}
\]
and hence 
\[
\sigma=\frac{2\pi rT^{*}}{T^{*}}=2\pi r.
\]

\subsection*{Computing $\protect\DD$}

We have that
\[
\pi^{*}(T^{*})=\TT^{*},
\]
and
\[
*\TT^{*}=g^{11}\det(g)^{1/2}\HH^{*}=r(1-\kappa r\cos(\theta))^{-1}\HH^{*},
\]
so that
\[
\pi_{*}(*\TT^{*})=\int_{0}^{2\pi}\left(\frac{r}{1-\kappa r\cos(\theta)}\right)d\theta=\frac{2\pi r}{1+r\kappa}\left(\sqrt{\frac{2}{1-r\kappa}-1}\right)T^{*}.
\]
Hence
\begin{eqnarray*}
\DD & = & \frac{D_{0}}{(1+r\kappa)}\left(\sqrt{\frac{2}{1-r\kappa}-1}\right).\\
 & = & D_{0}\sqrt{\frac{1}{1-r^{2}\kappa^{2}}}
\end{eqnarray*}

\section{Conclusions}

\label{conclusions}

We have shown that the diffusion equation on the total space of a
fiber bundle can be projected onto a diffusion equation on its base
space, under the hypothesis of infinite diffusion rate along the fibers.
We provided a general formula for the effective diffusion endomorphism
of the reduced diffusion equation, that we later applied to obtain
explicit formulas for diverse fiber bundles. Of particular interest
was the computation of the effective diffusion endomorphism associated
to the interface of two equidistant surfaces in 3-dimensional space,
in terms of the principal curvatures of the base surface.

\section{Appendix }

\label{sec:Appendix}

\subsection*{The sharp and flat operators}

We can see a 1-form $\alpha$ as a vector field $\alpha^{\sharp}$
defined by 
\[
<\alpha^{\sharp},X>=\alpha(X),
\]
where $<,>$ is the metric of the space under consideration. Similarly,
a vector field $X$ can be seen as a 1-form $X^{\flat}$ defined by
\[
X^{\flat}(Y)=<X,Y>.
\]
For a local frame $X_{1},\ldots,X_{n}$ the metric can be expressed
as a symmetric matrix $g$ with coefficients
\[
g_{ij}=<X_{i},X_{j}>,
\]
and we will write
\[
g^{ij}=(g^{-1})_{ij}.
\]
Let $X^{1},\ldots,X^{n}$ be the 1-forms forming the dual frame to
$X_{1},\ldots,X_{n}$, so that $X^{i}(X_{j})=\delta_{j}^{i}$. For
a vector field 
\[
X=\sum_{i=1}^{n}a^{i}X_{i}
\]
we have that 
\[
X^{\flat}=\sum_{i=1}^{n}a_{i}X^{i}\where a_{i}=\sum_{j=1}^{n}g_{ij}a^{j}.
\]
For a 1-form 

\[
\alpha=\sum_{i=1}^{n}\alpha_{i}X^{i}
\]
we have that
\[
\alpha^{\sharp}=\sum_{i=1}^{n}\alpha^{i}X_{i}\where\alpha^{i}=\sum_{j=1}^{n}g^{ij}\alpha_{j}.
\]

\subsection*{Hodge star operator}

The Hodge star operator $*$ maps $l$-forms to $(n-l)$ forms, where
$n$ is the dimension of the space under consideration, and it is
defined so that for $l$-forms $\omega$ and $\eta$ we have that
\[
\omega\wedge(*\eta)=<\omega,\eta>\mu,
\]
where $\mu$ is the volume form of the metric. For monomials 
\[
\omega=\alpha_{1}\wedge\ldots\wedge\alpha_{l}\and\eta=\beta_{1}\wedge\ldots\beta_{l},
\]
where the $\alpha_{i}$'s and the $\beta_{j}$'s are 1-forms, we have
that
\[
<\omega,\eta>=\det(<\alpha_{i},\beta_{j}>)\where<\alpha_{i},\beta_{j}>=<\alpha_{i}^{\sharp},\beta_{j}^{\sharp}>.
\]
The $*$-operator satisfies the duality relation, where for an $l$-form
$\omega$ we have that
\[
**\omega=(-1)^{l(n-l)}\omega.
\]
If for a local frame of vector fields $X_{1},\ldots,X_{n}$ we have
the metric matrix $g_{ij}=\left\langle X_{i},X_{j}\right\rangle $,
then the coefficients $g^{ij}$ of $g^{-1}$ are
\[
g^{ij}=<X^{i},X^{j}>.
\]
The metric volume form is given by
\[
\mu=\det(g)^{1/2}X^{1}\wedge\ldots\wedge X^{n}.
\]
Using the above formulas we obtain 
\[
*(X^{1}\wedge\ldots\wedge X^{n})=\det(g)^{-1/2},
\]
and 
\begin{equation}
*X^{i}=\det(g)^{1/2}\sum_{j=1}^{n}g^{ij}\iota_{j}(X^{1}\wedge\ldots\wedge X^{n}),\label{eq:StarOneForm}
\end{equation}
where 
\[
\iota_{j}(X^{1}\wedge\ldots\wedge X^{n})=(-1)^{j+1}X^{1}\wedge\ldots\wedge X^{i-1}\wedge\hat{X^{j}}\wedge X^{i+1}\wedge\ldots\wedge X^{n},
\]
and the hat in $\hat{X^{j}}$ indicates that that term has been removed
as a factor in the above wedge product. 

\bibliographystyle{plain}
\bibliography{myBib}

\begin{thebibliography}{10}

\bibitem{fj:antipov}
Anatoly~E. Antipov, Alexander~V. Barzykin, Alexander~M. Berezhkovskii, Yurii~A.
  Makhnovskii, Vladimir~Yu. Zitserman, and Sergei~M. Aldoshin.
\newblock Effective diffusion coefficient of a brownian particle in a
  periodically expanded conical tube.
\newblock {\em Physical Review E}, 88(054101), 2013.

\bibitem{kn:bott}
R.~Bott and L.W. Tu.
\newblock {\em Differentiable Forms in Algebraic Topology}.
\newblock Number~82 in Graduate Texts in Mathematics. Springer Verlag, 1982.

\bibitem{fj:bradley}
R.M. Bradley.
\newblock Diffusion in a two-dimensional channel with curved midline and
  varying width.
\newblock {\em Physical Review E}, 80(061142), 2009.

\bibitem{fj:di-projection-diffusion}
L.~Dagdug and I.~Pineda.
\newblock Projection of two-dimensional diffusion in a curved midline and
  narrow varying width channel onto the longitudinal dimension.
\newblock {\em The Journal of Chemical Physics}, 137(024107), 2012.

\bibitem{fj:kp-Extendex-fj-variational}
P.~Kalinay and K.~Percus.
\newblock Extended fick-jacobs equation: Variational approach.
\newblock {\em Physical Review E}, 72(061203), 2005.

\bibitem{fj:kp-diffusion-projection}
P.~Kalinay and K.~Percus.
\newblock Projection of a two-dimensional diffusion in a narrow channel onto
  the longitudinal dimension.
\newblock {\em The Journal of Chemical Physics}, 122(204701), 2005.

\bibitem{fj:kp-fick-jacob-correction}
P.~Kalinay and K.~Percus.
\newblock Corrections to the fick-jacobs equation.
\newblock {\em Physical Review E}, 74(041203), 2006.

\bibitem{fj:aproximations}
P.~Kalinay and K.~Percus.
\newblock Aproximations to the generalized fick-jacobs equation.
\newblock {\em Physical Review E}, 78(021103), 2008.

\bibitem{fj:ogawa-surfaces}
N.~Ogawa.
\newblock Curvature-dependent diffusion flow on a surface with thickness.
\newblock {\em Physical Review E}, 81(061113), 2010.

\bibitem{fj:ogawa}
N.~Ogawa.
\newblock Diffusion in a curved cube.
\newblock {\em Physics Letters A}, 377:2465--2471, 2013.

\bibitem{fj:projdiff}
C.~Valero and R.~Herrera.
\newblock Projecting diffusion along the normal bundle of a plane curve.
\newblock {\em Journal of Mathematical Physics}, 55(053509), 2014.

\bibitem{fj:entropybarrierzwanzig}
Robert Zwanzig.
\newblock Diffusion past an entropy barrier.
\newblock {\em The Journal of Chemical Physics}, 96(10):3926--3930, 1992.

\end{thebibliography}

\end{document}